\documentclass[aps,nofootinbib,twocolumn]{revtex4-1}

\usepackage{multirow}
\usepackage{tikz}
\usepackage{setspace}
\usepackage{amsmath,amsthm,amsfonts}

\definecolor{myurlcolor}{rgb}{0,0,0.7}
\definecolor{myrefcolor}{rgb}{0.8,0,0}
\usepackage{hyperref}
\hypersetup{colorlinks,
linkcolor=myrefcolor,
citecolor=myurlcolor,
urlcolor=myurlcolor}

\theoremstyle{plain}
\newtheorem{thm}{Theorem}
\newtheorem{lem}[thm]{Lemma}
\newtheorem{prop}[thm]{Proposition}

\theoremstyle{definition}

\theoremstyle{remark}

\newcommand{\R}{\mathbb{R}}

\newcommand{\beq}{\begin{equation}}
\newcommand{\eeq}{\end{equation}}
\renewcommand{\sp}{\mathrm{sp}}

\begin{document}

\title{Nonlocality with less Complementarity}

\author{Tobias Fritz}
\email[]{tobias.fritz@icfo.es}
\affiliation{ICFO -- Institut de Ciencies Fotoniques, Mediterranean Technology Park, 08860 Castelldefels (Barcelona), Spain}

\date{\today}

\begin{abstract}
In quantum mechanics, nonlocality (a violation of a Bell inequality) is intimately linked to complementarity, by which we mean that consistently assigning values to different observables at the same time is not possible. Nonlocality can only occur when some of the relevant observables do not commute, and this noncommutativity makes the observables complementary. Beyond quantum mechanics, the concept of complementarity can be formalized in several distinct ways. Here we describe some of these possible formalizations and ask how they relate to nonlocality. We partially answer this question by describing two toy theories which display nonlocality and obey the no-signaling principle, although each of them does not display a certain kind of complementarity. The first toy theory has the property that it maximally violates the CHSH inequality, although the corresponding local observables are pairwise jointly measurable. The second toy theory also maximally violates the CHSH inequality, although its state space is classical and all measurements are mutually nondisturbing: if a measurement sequence contains some measurement twice with any number of other measurements in between, then these two measurements give the same outcome with certainty. 
\end{abstract}

\maketitle

\section{Introduction}
\label{intro}

Since the groundbreaking work of Bell~\cite{Bell}, it has been recognized that quantum theory cannot be completed to a theory with local hidden variables; this result is known as \emph{Bell's theorem}, or almost synonymously as \emph{quantum nonlocality}. Bell's theorem is usually proven by deriving inequalities for the correlations between observables located at spatially separated sites, which are satisfied for any theory having local hidden variables, but violated by many quantum-mechanical models. See~\cite{Shimony} or \cite{BY} for background on Bell's theorem, including a clear overview of properties of hidden variable theories and the precise assumptions needed for the proof of Bell's theorem.

The fact that Bell inequalities can be violated in many quantum-mechanical models stems from the fact that, due to non-commutativity, a simultaneous joint measurement of these observables is not possible; quantum observables display \emph{complementarity}. But what does this mean, exactly? And what happens beyond quantum mechanics? From a theory-independent perspective, how do nonlocality and complementarity relate? This is the kind of question we are concerned with in this paper.

Recent work by Oppenheim and Wehner~\cite{OW} has brought to light a rather general theory-independent quantitative relationship between uncertainty relations and nonlocality. The absence of an uncertainty relation rules out nonlocality. They have also considered a certain different notion of complementarity\footnote{In Section D of~\cite{OW}. Note that Oppenheim and Wehner do not regard the existence of an uncertainty relation as complementarity.} (our property~\ref{adm}) and noted that the absence of such complementarity also rules out nonlocality.

In this paper, we would like to define two other concepts of complementarity. In the particular case of quantum theory, both are equivalent to non-commutativity of observables. For each of these two notions of complementarity, we then ask whether there are no-signaling theories which display nonlocality, although they do not display this kind of complementarity. We answer this in the positive by finding unphysical toy theories which have all the required properties.

\section{What is Complementarity?}
\label{wic}

So what actually is complementarity? In particular, once we leave the framework of quantum mechanics, what does it mean to say that observables $A$ and $B$ are complementary? Bohr~\cite{Bohr} has coined this term and used it to refer to the practical impossibility of a joint measurement of $A$ and $B$:
\begin{quote}
``[\ldots] it is only the mutual exclusion of any two experimental procedures, permitting the unambiguous
definition of complementary physical quantities,
which provides room for new physical laws, the
coexistence of which might at first sight appear
irreconcilable with the basic principles of science.
It is just this entirely new situation as regards
the description of physical phenomena, that the
notion of complementarity aims at characterizing.''
\end{quote}

This reduces the problem of defining complementarity to the problem of defining joint measurability.

Clearly, the concept of joint measurability should be taken to mean joint measurability \emph{in theory} rather than joint measurability \emph{in practice}; for if it were to mean the latter, than this would let the complementarity of $A$ and $B$ depend on the current state of the art in experiment and on the skill of the experimenter.

One can attribute different operational meanings to the concept of joint measurability of $A$ and $B$. We will consider the following four:

\bgroup
\renewcommand\theenumi{(\alph{enumi})}
\renewcommand\labelenumi{\theenumi}
\begin{enumerate}
\item\label{jd} Joint distribution: there is an observable $C$ from the measurement of which one can deduce both the value of $A$ and the value of $B$ (meaning that the probability distribution over outcomes of $C$ contains those of $A$ and $B$ as marginals).
\item\label{sdm} Symmetric nondisturbance of measurements: in the measurement sequence $ABA$, the two measurements of $A$ always give the same result. Similarly for the sequence $BAB$.
\item\label{adm} Asymmetric nondisturbance of measurements~\cite[{Sec.~7-4.$\mathcal{C}$}]{Peres}: on any initial state, the measurement sequence $AB$ gives the same probability distribution over outcomes of $B$ as a direct measurement of $B$.
\item\label{ur} No uncertainty relation: there is no non-trivial uncertainty relation between the outcome distribution of $A$ and that of $B$. This can be formalized in several inequivalent ways, e.g. in terms of entropic uncertainty relations~\cite{MU} or in terms of the fine-grained uncertainty relations of~\cite{OW}.
\end{enumerate}

\begin{lem}
\label{qc}
In the case of quantum theory with projective measurements, properties~\ref{jd},~\ref{sdm} and~\ref{adm} are equivalent to commutativity of $A$ and $B$.
\end{lem}

This is straightforward to prove, but we nevertheless include a proof sketch for the sake of completeness.

\begin{proof}
We work in terms of the spectral decompositions $B=\sum_{\lambda\in\sp(A)} \lambda P_\lambda$ and $B=\sum_{\lambda\in\sp(B)} \lambda Q_\lambda$. Each $P_\lambda$ can be written as a polynomial in $A$, and each $Q_\lambda$ as a polynomial in $B$. So if $A$ and $B$ commute, then $P_\lambda Q_\mu=Q_\mu P_\lambda$ for all $\lambda,\mu$; the converse is also clear.

To see that property~\ref{jd} follows from commutativity, choose an injection $j:\sp(A)\times\sp(B)\hookrightarrow\R$ and define $C=\sum_{\lambda\in\sp(A),\,\mu\in\sp(B)} j(\lambda,\mu) P_\lambda Q_\mu$. By injectivity of $j$, a measurement of $C$ automatically also measures both $A$ and $B$. Conversely, if $C$ corresponds to a joint measurement of $A$ and $B$ as in property~\ref{jd}, then there are polynomials $f$ and $g$ such that $A=f(C)$ and $B=g(C)$, and therefore $A$ commutes with $B$.

Property~\ref{sdm} and property~\ref{adm} are each equivalent to $\sum_\lambda P_\lambda Q_\mu P_\lambda = Q_\mu$ and $\sum_\mu Q_\mu P_\lambda Q_\mu = P_\lambda$. The first equation implies that $P_\lambda Q_\mu P_\lambda'=0$ for $\lambda\neq\lambda'$. Therefore,
\begin{align*}
P_\lambda Q_\mu & = \sum_{\lambda'} P_\lambda Q_\mu P_{\lambda'} = \sum_\lambda P_\lambda Q_\mu P_\lambda \\
& = \sum_{\lambda'} P_{\lambda'} Q_\mu P_\lambda = Q_\mu P_\lambda
\end{align*}
which implies commutativity of $A$ and $B$. The same calculation also shows the converse implication.
\end{proof}

Despite this equivalence in the quantum case, beyond quantum theory these four notions of complementarity are conceptually very different. For example, properties~\ref{jd} and~\ref{ur} do not care about any post-measurement states under measuring either $A$ or $B$, but only about the distribution of outcome probabilities; for property~\ref{adm}, this only applies to measurement $A$; while for property~\ref{sdm}, the post-measurement states of both $A$ and $B$ are relevant. Further discussion on the interrelations between these properties is beyond the scope of this work. 

\section{Nonlocality requires Complementarity?}

It seems to be intuitive that complementarity is a necessary requirement for displaying nonlocality. One may believe that any no-signaling theory which does not display complementarity admits a local realistic description. Indeed, this can be made precise for property~\ref{adm} as follows~\cite[D.1]{OW}: if property~\ref{adm} holds for all pairs $A$, $B$, then it follows that the outcome distribution of any observable $A$ does not depend on whether other $A$ is measured directly, or other observables are measured before $A$. This implies that one can assign a joint probability distribution to all the observables at once, and hence the existence of a local realistic model of the resulting correlations. For bounds on certain kinds of nonlocality in terms of uncertainty relations as in property~\ref{ur}, we also refer to~\cite{OW}.

One might now conjecture that also the properties~\ref{jd} and~\ref{sdm}, together with the no-signaling principle, are likewise sufficient in order to exlucde nonlocal behavior.

In the following, we will show that this is not necessarily so. In Section~\ref{toy1} (resp.~\ref{toy2}), we will describe a no-signaling toy theory which has property~\ref{jd} (resp.~\ref{sdm}) for every pair of measurements $A,B$, but maximally violates the Bell inequality of Clauser, Horn, Shimony and Holt~\cite{CHSH} (CHSH inequality). These two toy theories are more classical than quantum mechanics in the sense of having less complementarity, but nevertheless display a higher degree of nonlocality.

There are several things that should be kept in mind while reading this paper. First and most importantly, our constructions are \emph{toy theories}, which means that we do not ascribe any physical significance to them. In fact, they can easily be seen to be unphysical. Any toy theory displaying nonlocality with ``less'' complementarity than quantum mechanics is necessarily unphysical---precisely \emph{because} of the very feature of displaying less complementarity than quantum mechanics, which is our current theoretical framework for (microscopic) physics. Therefore, the present investigations should not be regarded as actual physics in the sense of describing reality, but rather as theoretical investigations around the foundations of Bell's theorem.

Second, the systems which make up our toy theories are nonlocal systems, in the sense that their states and observables are only defined globally in the whole ``toy universe'', and not locally. This is just like in quantum mechanics, where wave functions are likewise nonlocal entities. For simplicity, our small toy universe only consists of two independent worldline segments called ``Alice'' and ``Bob'', which we take to be spacelike separated.

Third, what we regard as a ``theory'' comprises the definition of \emph{both} states and measurements. While in some frameworks, the set of measurements may be determined by the set of states (and/or vice versa), we do not assume this to be the case. Our toy theories can be formalized in the language of generalized probabilistic theories and then satisfy the Assumptions 1, 2, 3 and 7 of~\cite{B}, but not its Assumptions 4, 5 and 6. See also footnote 4 and Section VIII.A in~\cite{B}.

Fourth, while our toy theories are obviously tailored to achieve maximal nonlocality in the CHSH scenario, it makes perfect sense to consider them independently of any considerations involving Bell inequalities or nonlocality, and in particular independently of any particular Bell scenario. 

Sections~\ref{toy1} and~\ref{toy2} can be read independently.

\section{Nonlocality with jointly measurable observables}
\label{toy1}

We now describe the first toy theory which maximally violates the CHSH inequality by reproducing the correlations of the Popescu-Rohrlich box (PR-box~\cite{PR}, Table~\ref{PRbox}), but nevertheless satisfies property~\ref{jd} for all relevant pairs of observables. The basic idea is that the PR-box correlations between Alice's observables $A_1,A_2$ and Bob's observables $B_1,B_2$ do not prevent $A_1$ and $A_2$ from having a joint distribution; and likewise for $B_1$ and $B_2$. (What it \emph{does} prevent is the existence of a joint distribution for \emph{all four} observables together~\cite{F}.) Moreover, for each Alice and Bob one can define an observable which probes this joint distribution.

We will define the state space of the toy theory in terms of four ``basic observables'' $A_1,A_2,B_1,B_2$. We assume these four observables to be $\pm 1$-valued. Following the idea outlined in the previous paragraph, we now define a (mixed) state of the toy theory to be an assignment of joint outcome probabilities to each pair of observables, such that the different ways to calculate the marginal of each single observable all give the same two-outcome distribution. More concretely, such a state is given by six four-outcome probability distributions
\begin{align}
\begin{split}
\label{sixdist}
P_{A_1,A_2}(a_1,a_2) &,\quad P_{A_1,B_1}(a_1,b_1), \quad P_{A_1,B_2}(a_1,b_2),\\
P_{A_2,B_1}(a_2,b_1) &,\quad P_{A_2,B_2}(a_2,b_2), \quad P_{B_1,B_2}(b_1,b_2),
\end{split}
\end{align}
satisfying the marginal conditions
$$
\sum_{y} P_{X,Y}(x,y) = \sum_{y} P_{X,Y'}(x,y)
$$
for all observables $X,Y,Y'\in\{A_1,A_2,B_1,B_2\}$ and any outcome $x$. When $X$ is an observable of Alice and $Y,Y'$ are Bob's, or vice versa, these are precisely the no-signaling equations.

Since the set of all these states is convex, it is automatically closed under probabilistic mixtures. In the abstract framework of~\cite{AbramBrand}, this state space is the set of empirical models over four measurements $A_1,A_2,B_1,B_2$ with measurement contexts precisely all the pairs of measurements.

This definition should make it clear that we would like to regard the four basic observables $A_1,A_2,B_1,B_2$ as measurements which can all be performed pairwise jointly, but not triplewise jointly. Following this idea, the set of observables of the toy theory is defined to consist of the basic observables together with the four-outcome observables representing pairwise joint measurements. 

The physical picture behind joint measurability of all observable pairs is as follows: there is a source of entangled particle pairs which distributes them to Alice's and Bob's labs, respectively. In Alice's lab, she can choose to measure either $A_1$ or $A_2$; in Bob's lab, he can choose between $B_1$ and $B_2$. Moreover, as an additional option, there is a third experimenter, who has access to the source of entangled pairs; let us call him Charlie. Charlie can decide to let an entangled pair pass by and reach the labs of Alice and Bob; or he can decide to apply a measurement of one of two four-outcome observables $C_A$ and $C_B$. In this latter case, the entangled pair gets destroyed due to the measurement process. The three experimenters find out that the observable $C_A$ behaves like a joint distribution of $A_1$ and $A_2$, while $C_B$ behaves like a joint distribution of $B_1$ and $B_2$, in the sense of property~\ref{jd}. 

This toy theory has the following features:

\begin{table}
\begin{center}
\begin{tabular}{cc|cc|cc}
\multicolumn{6}{c}{}\\
 \multicolumn{2}{c|}{} & \multicolumn{2}{|c|}{$A_1$} & \multicolumn{2}{|c}{$A_2$}\\
\multicolumn{2}{c|}{} & \multicolumn{1}{|c}{$+1$} & \multicolumn{1}{c|}{$-1$} & \multicolumn{1}{|c}{$+1$} & \multicolumn{1}{c}{$-1$}\\
\hline \multirow{2}{*}{$B_1$} & \multicolumn{1}{c|}{$+1$} & \multicolumn{1}{|c}{$\frac{1}{2}$} & \multicolumn{1}{c|}{$0$} & \multicolumn{1}{|c}{$\frac{1}{2}$} & \multicolumn{1}{c}{$0$} \\
 & \multicolumn{1}{c|}{$-1$} & \multicolumn{1}{|c}{$0$} & \multicolumn{1}{c|}{$\frac{1}{2}$} & \multicolumn{1}{|c}{$0$} & \multicolumn{1}{c}{$\frac{1}{2}$} \\
\hline \multirow{2}{*}{$B_2$} & \multicolumn{1}{c|}{$+1$} & \multicolumn{1}{|c}{$\frac{1}{2}$} & \multicolumn{1}{c|}{$0$} & \multicolumn{1}{|c}{$0$} & \multicolumn{1}{c}{$\frac{1}{2}$} \\
 & \multicolumn{1}{c|}{$-1$} & \multicolumn{1}{|c}{$0$} & \multicolumn{1}{c|}{$\frac{1}{2}$} & \multicolumn{1}{|c}{$\frac{1}{2}$} & \multicolumn{1}{c}{$0$} \\
\end{tabular}
\end{center}
\caption{Table of joint outcome probabilities for the Popescu-Rorhlich box.}
\label{PRbox}
\end{table}

\begin{itemize}
\item Among the basic observables $A_1,A_2,B_1,B_2$, all pairs are jointly measurable in the sense of property~\ref{jd};
\item The theory contains states which display PR-box correlations: defining a state through the joint distributions of Table~\ref{PRbox} together with any joint distribution of $A_1$ and $A_2$ with uniform marginals, and likewise between $B_1$ and $B_2$, gives the desired result.
\end{itemize}

We conclude that joint measurability of local observables together with the no-signaling principle is not sufficient to guarantee the existence of a local hidden variable description.

\section{Nonlocality with mutually nondisturbing measurements}
\label{toy2}

We will now define our second toy theory in terms of a finite number of states and observables. Similar to before, the observables are all $\pm 1$-valued and they are going to be denoted by $A_1$, $A_2$, $B_1$, $B_2$; we will think of $A_1$ and $A_2$ as observables located on the subsystem ``Alice'', while $B_1$ and $B_2$ are observables on the subsystem ``Bob''. The model will have the following properties:

\begin{enumerate}
\item All measurements are perfectly mutually nondisturbing: sequential measurements of the same observable always give the same value, no matter which other measurements were conducted in the meantime. \\
(On the other hand, for each measurement, there are some states on which the outcome is random. Post-measurement states are in general different from pre-measurement states.)
\item Signaling is impossible: the outcome probability distribution of any measurement sequence of Bob does not depend on the actions of Alice, and vice versa.
\item The correlations between $A_i$ and $B_j$ constitute a Popescu-Rohrlich box~\cite{PR}, thereby maximally violating the CHSH inequality~\cite{CHSH}.
\end{enumerate}

Before diving into the details, we describe a different model\footnote{due to Matthew Pusey.} which is operationally almost equivalent to our toy theory. Suppose that Alice and Bob share a PR-box, and in addition each party has a memory of two bits available. The purpose of Alice's two bits is to store any previously measured values of $A_1$ and $A_2$, and likewise for Bob's two bits. Now upon the first measurement carried out on each subsystem, the PR-box is used, and the measurement outcome is stored in the corresponding bit, while the other local bit is initialized randomly and independently. All subsequent measurements only reproduce the values of the corresponding local bit. This simple protocol already is a model having property~\ref{sdm} while displaying PR-box correlations.

The following toy theory is essentially an elaboration on this idea. It has the additional feature of having a classical state space. Also, it does not artificially distinguish the first measurement from subsequent ones.

A pure state in the toy theory is defined to be a formal expression of the form $a_1^j a_2^k b_1^l b_2^m$, where each of the upper indices is an element of the set of \emph{hidden values}
$$
V=\{\emptyset , -, +, \ominus, \oplus\}\:.
$$
An example of a pure state is $a_1^+a_2^\emptyset b_1^\ominus b_2^-$. The elements of $V$ correspond to properties of the associated observable as follows. When the state is of the form $a_1^\emptyset\ldots$, then we say that the hidden value of the observable $A_1$ is \emph{undetermined}; on states of the form $a_1^-\ldots$, we say that $A_1$ has a \emph{potential value} of $-1$, while on those which look like $a_1^+\ldots$, the observable $A_1$ has a potential value of $+1$. On states like $a_1^\ominus\ldots$, we associate to $A_1$ an \emph{actual value} of $-1$, while on $a_1^\oplus\ldots$, the observable $A_1$ has an actual value of $+1$. The same applies to the observables $A_2$, $B_1$ and $B_2$. We will soon explain the significance of undetermined values, potential values and actual values.

A mixed state is defined to be a probabilistic ensemble of pure states. This means that the theory has a classical state space in the sense that every state is a \emph{unique} probabilistic mixture of pure states. We will mostly work on the level of pure states rather than ensembles; when not explicitly mentioned otherwise, ``state'' means ``pure state''.

We now need to define the measurements. Since now we are interested in measurement sequences, we also need to consider post-measurement states. Therefore, a measurement definition consists of assignments of probabilities for each outcome on each pre-measurement state, as well as a (probabilistic) assignment of post-measurement state for each combination of outcome and pre-measurement state which occurs with non-zero probability. We take these data to be given as follows:

\begin{table}
\centering
\begin{tabular}{c|c|c}
 & outcome $-1$ & outcome $+1$ \\
\hline
$A_1$ & $a_1^\ominus a_2^\emptyset b_1^- b_2^- $ & $a_1^\oplus a_2^\emptyset b_1^+ b_2^+ $ \\
$A_2$ & $a_1^\emptyset a_2^\ominus b_1^- b_2^+ $ & $a_1^\emptyset a_2^\oplus b_1^+ b_2^- $ \\
$B_1$ & $a_1^- a_2^- b_1^\ominus b_2^\emptyset $ & $a_1^+ a_2^+ b_1^\oplus b_2^\emptyset $ \\
$B_2$ & $a_1^- a_2^+ b_1^\emptyset b_2^\ominus $ & $a_1^+ a_2^- b_1^\emptyset b_2^\oplus $ 
\end{tabular}
\caption{Table of post-measurement states for the model described in Section~\ref{toy2} with pre-measurement state $a_1^\emptyset a_2^\emptyset b_1^\emptyset b_2^\emptyset$. For each observable, each outcome has probability $1/2$, and the post-measurement state is given by the table entry corresponding to that outcome.}
\label{meas}
\end{table}

\begin{enumerate}
\item The outcome probabilities for any observable are determined by the hidden value of that observable. If this value is undetermined, both outcomes occur equally likely with probability $1/2$; otherwise, the potential value or actual value is reproduced with certainty.
\item The post-measurement states are defined by specifying how the hidden values of the observables change. For all pre-measurement states except $a_1^\emptyset a_2^\emptyset b_1^\emptyset b_2^\emptyset$, we define this to work as follows. The hidden value of the observable always turns into that actual value which corresponds to the measurement outcome. All other hidden values stay untouched, except when the other observable of the same party has a potential value; in this case, this potential value flips its sign with a probability of $1/2$.
\item The post-measurement states of $a_1^\emptyset a_2^\emptyset b_1^\emptyset b_2^\emptyset$ are as in Table~\ref{meas}.
\end{enumerate}

For example, on the pre-measurement state $a_1^+a_2^\emptyset b_1^\ominus b_2^-$, a measurement of $A_1$ gives a definite $+1$ outcome and yields the post-measurement state $a_1^\oplus a_2^\emptyset b_1^\ominus b_2^-$; a measurement of $A_2$ gives a random outcome; if this outcome is $+1$, then the post-measurement state is randomly chosen between $a_1^+ a_2^\ominus b_1^\ominus b_2^-$ and $a_1^- a_2^\ominus b_1^\ominus b_2^-$; if it is $-1$, then between $a_1^+ a_2^\oplus b_1^\ominus b_2^-$ and $a_1^- a_2^\oplus b_1^\ominus b_2^-$. Measuring $B_1$ gives a definite $-1$ outcome where the two post-measurement states $a_1^+a_2^\emptyset b_1^\ominus b_2^-$ and $a_1^+a_2^\emptyset b_1^\ominus b_2^+$ are equally likely. Finally, a measurement of $B_2$ has a definite $-1$ outcome, with post-measurement state $a_1^+a_2^\emptyset b_1^\ominus b_2^\ominus$.

This ends the definition of the model. We can now turn to the study of its properties. 

One characteristic property of classical systems is that measurements do not change the state of the system. This is not always the case in our model; for example, when the pre-measurement state is $x=a_1^\emptyset a_2^\emptyset b_1^\emptyset b_2^\emptyset$, then the post-measurement is always different from $x$, no matter what the measurement and its outcome are. Nevertheless, the model has several nice properties which make it look classical in some respects. One of these is the following:

\begin{lem}
\label{mutually nondisturbing}
For any pre-measurement state, all measurements in this model are mutually nondisturbing: if
$$
\ldots,\, A_i,\, \ldots,\, A_i,\, \ldots
$$
is any sequence of measurements containing $A_i$ at least twice, then these two measurements of $A_i$ have the same outcome with probability $1$. Likewise for $B_j$.
\end{lem}

\begin{proof}
By definition of post-measurement states, after the first measurement of $A_i$ the state has an actual value on $a_i$, and this value never changes in any subsequent measurement. Hence all further measurements of $A_i$ will reproduce this value with certainty.
\end{proof}

\begin{prop}
\label{nosig}
The model does not allow signaling in the following sense: for any initial state, the outcome distribution of any measurement sequence of Bob does not depend on how many and which measurements Alice conducts in the meantime, and conversely.
\end{prop}

\begin{proof}
On any initial state which is not $a_1^\emptyset a_2^\emptyset b_1^\emptyset b_2^\emptyset$, this is clear since the measurement rules are ``local'' in the sense that measuring $A_1$ or $A_2$ can only change the hidden values of $a_1$ and $a_2$, and likewise for $B_1$ or $B_2$.

On the initial state $a_1^\emptyset a_2^\emptyset b_1^\emptyset b_2^\emptyset$, one can reason as follows. We may assume without loss of generality that Bob's measurement sequence contains both $B_1$ and $B_2$ at least once. Then due to lemma~\ref{mutually nondisturbing}, it is actually enough to assume that he measures both of them exactly once. Then we claim that all four outcome sequences $(\pm 1,\pm 1)$ occur with an equal probability of $1/4$ in all cases. This is easy to verify if the first measurement is conducted by Bob, since then again Alice's actions do not change any hidden values of Bob's. If the first measurement is conducted by Alice, then by Table~\ref{meas}, Bob's hidden values are both potential. They are perfectly correlated if Alice's first measurement was $A_1$, and perfectly anticorrelated if she started with $A_2$. In either case, this correlation gets erased due to Bob's first measurement, which may flip one of the two potential values, so that all four possible sign combinations are equally likely, independently of what Alice did.
\end{proof}

For this last part of the proof, the distinction between potential values and actual values is crucial: if we would regard all potential values as actual values which never change in any measurement, then signaling would be possible since Bob could measure whether $B_1$ and $B_2$ are perfectly correlated or perfectly anticorrelated, and thereby he would know whether Alice measured $A_1$ or $A_2$.

\begin{prop}
On the initial state $a_1^\emptyset a_2^\emptyset b_1^\emptyset b_2^\emptyset$, the correlations between $A_i$ and $B_j$ coincide with those of the PR-box (Table~\ref{PRbox}).
\end{prop}

\begin{proof}
This can be directly checked from the rules specified above together with Table~\ref{meas}, which has been constructed precisely in such a way that the correlations of the Popescu-Rohrlich box (Table~\ref{PRbox}) can be reproduced.
\end{proof}

\paragraph*{Reply to potential criticism.} There are several deficiences which one may deem this toy theory to have and which have been pointed out to us in discussion. We would like to address some of these now.

First, one may remark that the theory is unphysical: it has almost trivial dynamics and contradicts quantum mechanics. This is certainly correct, and it ought to be kept in mind that this is the whole point of this investigation: how classical can a theory be while still displaying nonlocality?

Second, one may object that the theory is very unnatural: the state $a_1^\emptyset a_2^\emptyset b_1^\emptyset b_2^\emptyset$ is a state that can never be prepared within the theory starting from any other state. While this is correct, one can easily amend the theory by extra bells and whistles which remedy this problem and similar ones. For example, one can add a new many-outcome observable $X$ which randomly prepares any state and whose outcome is the numerical representation of this prepared state.

Third, one may wonder what would happen when Alice and Bob measure simultaneously? It seems that joint measurements have not been defined in the theory. But thanks to Proposition~\ref{nosig}, this is irrelevant. If Alice measures $A_i$ and Bob simultaneously measures $B_j$, one can describes this in the toy theory as either the measurement sequence $A_iB_j$ or as the measurement sequence $B_jA_i$. While these two time orderings are inequivalent as transformations on states, Proposition~\ref{nosig} implies that the observational predictions coincide. 

\section{Discussion}

The notion of complementary in fundamental physics refers to the phenomenon that to some pairs of observables, it is impossible to consistently assign values to both of them at the same time. It is generally believed that a no-signaling theory without complementary observables cannot display nonlocality. Along these lines, a quantitative relationship between uncertainty relations and certain kinds of nonlocality has been worked out in~\cite{OW}. In the work, we have pointed out that nonlocality does not necesssarily require complementarity.

In Section~\ref{wic} we have discussed several ways to formalize the notion of joint measurability as properties~\ref{jd} to~\ref{ur}, and noted that a no-signaling theory with property~\ref{adm} or~\ref{ur} for all pairs of observables always allows local hidden variables and therefore cannot display nonlocality. In Section~\ref{toy1}, we have discussed a simple toy example of a generalized probabilistic theory which maximally violates the CHSH inequality by displaying PR-box correlations, although the corresponding local observables are jointly measurable in the sense of property~\ref{jd}. In Section~\ref{toy2}, we have described a no-signaling toy theory which also displays PR-box behavior, although its state space is classical and it satisfies property~\ref{sdm} for all pairs of observables.

While we expect our results to be valid not only for maximal violations of the CHSH inequality, but more generally for any no-signaling violation of any Bell inequality, we have not considered these more general cases so far. For the sake of maximal concreteness, we have preferred to consider a specific example only.

\section*{Acknowledgements}

The author would like to thank Manuel B\"arenz for raising a question which lead to the present results and for valuable feedback on a draft; several critical referees for the same; Antonio Ac\'in, Andrei Khrennikov, Anthony Leverrier and Masanao Ozawa for stimulating discussions and various crucial suggestions; Jean-Daniel Bancal, Gonzalo de la Torre, Giuseppe Prettico and Alessio Serafini for further discussion; and the EU STREP QCS for financial support.

\bibliographystyle{plain}
\bibliography{bell}

\end{document}